\pdfoutput=1
\documentclass[11pt]{llncs}
\usepackage{ifthen}              % conditional text
\newboolean{HIDE}                % declaration of Boolean variable
\setboolean{HIDE}{true}          % assign value true to show  hidden text
\usepackage[letterpaper,hmargin=1.1in,vmargin=1.23in]{geometry}
\usepackage{graphicx}            % graphics import
\usepackage{amssymb}             % the inevitable AMS packages
\usepackage{amsmath}
\usepackage{dsfont}
\usepackage{lmodern}             % high quality fonts
\usepackage{longtable}           % tables over more pages possible
\newcommand{\F}{\mathbb{F}}             % The Galois field F
\newcommand{\xor}{\oplus}               % the XOR
\newcommand{\dmin}{d_\text{min}}        % min. code distance
\newcommand{\V}[1]{\mathbf{#1}}         % denote a vector (roman bold, not slanted)
\newcommand{\PM}{\begin{pmatrix}}       % begin and end of an AMS matrix
\newcommand{\MP}{\end{pmatrix}}
\newcommand{\BE}{\begin{equation}}      % begin and end of equations
\newcommand{\EE}{\end{equation}}
\newcommand{\BEA}{\begin{eqnarray*}}    % begin and end of unnumbered equation arrays
\newcommand{\EEA}{\end{eqnarray*}}
\newcommand{\BDEF}{\begin{definition}}  % begin and end of definitions
\newcommand{\EDEF}{\end{definition}}
\newcommand{\eg}{{e.g.\ }}
\newcommand{\ie}{{\it i.e.\ }}

\newcommand{\cf}{{\it cf.\ }}
\newcommand{\QED}{\hfill $\Box$}
\newboolean{COMMIT}                      % declaration of Boolean variable
\setboolean{COMMIT}{true}                % commit changes (true) else set to (false)
\ifCOMMIT
 
 \newcommand{\delpar}[1]{}
 \newcommand{\delw}[1]{}
 
\else
\usepackage{color}
\definecolor{bgnd}{gray}{.4}
\newcommand{\delpar}[1]{\

\colorbox{bgnd}{\parbox{\textwidth}{#1}}\newline}             % delete paragraph
\newcommand{\delw}[1]{\colorbox{bgnd}{#1}}                    % delete words
\fi
\begin{document}
%%%%%%%%%%%%%%%%%%%%%%%%%%%%%%%%%%%%%%%%%%%%%%%%%%%%%%%%%%%%%%%%%%%%%%%%%%%%%%%%
%%%%%%%%%%%%%%%%%%%%%%%%%%%%%%%%%%%%%%%%%%%%%%%%%%%%%%%%%%%%%%%%%%%%%%%%%%%%%%%%
%
\title{On the Duality of Probing and Fault Attacks}
\author{%
\ifHIDE
  Berndt M. Gammel and Stefan Mangard
\else
  \emph{anonymous authors}\newline
\fi%
}
\ifHIDE
 \institute{Infineon Technologies AG \\
            Munich, Germany\\[1em]
            Feb. 12, 2009\\[1em]
            \email{Berndt.Gammel@infineon.com}\\
            \email{Stefan.Mangard@infineon.com}}
\else
 \institute{\quad\vspace*{8em}}
\fi
\maketitle
%
%%%%%%%%%%%%%%%%%%%%%%%%%%%%%%%%%%%%%%%%%%%%%%%%%%%%%%%%%%%%%%%%%%%%%%%%%%%%%%%%

\begin{abstract}
In this work we investigate the problem of simultaneous privacy and
integrity protection in cryptographic circuits. We consider a
white-box scenario with a powerful, yet limited attacker. A concise
metric for the level of probing and fault security is introduced,
which is directly related to the capabilities of a realistic
attacker. In order to investigate the interrelation of probing and
fault security we introduce a common mathematical framework based on
the formalism of information and coding theory. The framework
unifies the known linear masking schemes. We proof a central theorem
about the properties of linear codes which leads to optimal secret
sharing schemes. These schemes provide the lower bound for the
number of masks needed to counteract an attacker with a given
strength. The new formalism reveals an intriguing \emph{duality
principle} between the problems of probing and fault security, and
provides a unified view on privacy and integrity protection using
error detecting codes. Finally, we introduce a new class of linear
tamper-resistant codes. These are eligible to preserve security
against an attacker mounting simultaneous probing and fault attacks.

\vspace*{1em}
\end{abstract}

%%%%%%%%%%%%%%%%%%%%%%%%%%%%%%%%%%%%%%%%%%%%%%%%%%%%%%%%%%%%%%%%%%%%%%%%%%%%%%%%
\section{Introduction}   \label{Sec:Introduction}

In the traditional cryptographic setting it is assumed that
the adversary has only \emph{black-box} access to the cryptographic algorithm.
He can query the apparatus executing
the algorithm with inputs of his choice and
observe the answer (chosen plain text scenario).
The secret key has been loaded into
the device in the outset and is not accessible to the adversary.
A further basic assumption is that the attacker has full knowledge of
the algorithm---he can build a model of the device
and query it based on his guess for the secret.

In practice, however, the black-box assumption is realized rarely.
This applies particularly to a device which can be physically
accessed by the attacker during operation. Present-day society
heavily relies on the security of such cryptographic devices.
Examples are electronic ID cards, smart cards for payment purposes,
mobile phones, network devices, and personal computers. All of these
can be easily seized by an attacker. In this setting both the
\emph{privacy} and the \emph{integrity} assumption inherent in the
black-box model do not hold any more:

Firstly, due to the physical nature of computation there is always
some amount of information leakage from intermediate stages of the computation.
Secondly, information from intermediate stages can be actively probed
without disturbing the computation.
Thirdly, the physical processes taking place during computation
can be actively disturbed in order to induce faulty intermediate values.
These three fundamental physical constraints on cryptographic computation
are known as
the \emph{side-channel} (SCA),
the \emph{probing} (PRA), and
the \emph{fault} (FA)
attack scenario, respectively.

There are many physical sources for side-channel information
leakage during the execution of an algorithm
which is otherwise secure under the black-box assumption.
Especially the analysis of the information leakage on
the power line of a physical device has had considerable
impact on the art of secure cryptographic algorithm implementation
during the last decade~\cite{Kocher1999DifferentialPowerAnalysis}.
Differential Power Analysis (DPA) represents an actual threat
for commonly used cryptographic devices, because the
work factor for a DPA attack is comparably moderate
with respect to the equipment, the skill, and the computing power
\cite{Chari1999TowardsSoundApproaches,%
      Mangard2007PowerAnalysisAttacks,%
      Messerges2000SecuringtheAES}.

Probing attacks can be considered as even more powerful
than SCA attacks, because the attacker
monitors a \emph{local} physical value (e.g. a voltage level on some wire),
which is directly related to a data value.
There are many sophisticated probing techniques available
ranging from the placement of needles to optical probing methods
\cite{Anderson1996TamperResistance-,%
%      Anderson1998SoftTempestHiddenData,%
      Boit2008PhysicalTechniquesChip-Backside}.
Throughout this paper we will use the term \emph{probe}
for any method that allows to record
a local value in a computation.

Both SCA and PRA attacks are highly efficient if they are used
in the differential setup. Here, a few carefully chosen local values
are probed and traces for several runs of the algorithm
with different inputs are recorded.
The analysis of the collected data values can already reveal
the whole key or at least reduce the key space
to a level that allows a brute force attack.
Theoretical results on highly efficient differential probing attacks (DPRA) on
private key and public key cryptosystems have recently been reported
\cite{Handschuh1999ProbingAttackson,%
      Schmidt2009ProbingAttackAES}.

The third class of physical attacks, fault attacks, is
particularly interesting, because the violation of integrity
can be exploited to break privacy
\cite{Anderson1996TamperResistance-,%
      Skorobogatov2003OpticalFaultInduction}.
It has been demonstrated that the injection of a small
number of a specific kind of faults can be used to break
private key
\cite{Biham1997DifferentialFaultAnalysis}
and public key
\cite{Boneh2001ImportanceofEliminating}
cryptosystems.

Clearly, there is a tireless quest for countermeasures against all
of these kinds of attacks. As physical security seems to be out of
reach (in the sense that anything can be proven) countermeasures
based on mathematical reasoning are of particular interest. However,
in the field of probing attacks the impossibility of obfuscation
\cite{Barak2001Impossibilityofobfuscating} rules out the security
against an all-powerful attacker. Hence, the capabilities of the
attacker have to be restricted by model assumptions
\cite{Micali2004PhysicallyObservableCryptography}. \emph{Secret
sharing} schemes have been proposed to provide privacy against an
attacker who is limited to place at most $q$ needles
\cite{Ishai2003PrivateCircuits:Securing}. In the context of DPA
secret sharing schemes have been introduced as \emph{data masking}.
For example, in a simple masking scheme with two shares a data value
$x$ is lifted to the pair of values $(x \xor m, m)$, where the mask
value $m$ has the properties of a uniformly distributed random
variable. Obviously the pair $(x \xor m, m)$ is secure against
probing with one needle.

The detection (and correction) of faults has a
long history in coding theory
\cite{MacWilliams2006TheoryofError-Correcting,%
      Shannon1949CommunicationTheoryof}.
Trivially, the capability to detect errors requires the
introduction of information redundancy.
An appropriate error detection code can be devised depending on
the kind of errors the attacker is able to inject.
If, for example, it can be assumed that an attacker is able to
flip at most $f$ bits in a memory word, and the attack
should be detected with certainty, an error detection
code with minimum distance $f+1$ could be used.
On the other hand
error detecting codes have also be used for the construction of
secret sharing schemes~\cite{Massey1993MinimalCodewordsand}.

The crucial challenge is, however,
to provide security against an attacker of limited, but considerable power,
who is able to perform both probing and fault attacks.
First theoretical foundations for such \emph{tamper-resistant} devices
have been laid in recent works
\cite{Gennaro2004AlgorithmicTamperProof,%
      Ishai2006PrivateCircuitsII}.
Both, secret sharing schemes and error detection codes introduce
redundancy into the realization of the cryptographic device.
However, little is known about the interrelation between simultaneous probing and
fault resistance countermeasures.

This work is organized in four sections.
In Section~\ref{Sec:Preliminaries} we define a concise metric
for the level of probing security,
which is directly related to the capabilities of a realistic attacker.
In order to investigate the interrelation of probing security
and fault security we introduce a common mathematical language
within the framework of information and coding theory.
In Section~\ref{Sec:OPS} we prove a theorem about the
properties of codes which can be used for the construction
of optimal masking schemes.
In particular, we describe linear codes which
are optimal with respect to the number of introduced masks (OPS codes).
The new formalism reveals an intriguing \emph{duality principle}
between the problems of probing and fault security.
In Section~\ref{Sec:Leakage} we compare the
information leakage of optimal privacy preserving codes and
classical masking schemes.
Finally, in Section~\ref{Sec:OTR} we fuse privacy and integrity protection
and introduce a new class of
optimal tamper-resistant codes (OTR codes),
which are eligible to preserve security against an attacker mounting
simultaneous probing and fault attacks.

%%%%%%%%%%%%%%%%%%%%%%%%%%%%%%%%%%%%%%%%%%%%%%%%%%%%%%%%%%%%%%%%%%%%%%%%%%%%%%%%
\section{Preliminaries}  \label{Sec:Preliminaries}
It is suggestive to describe an arbitrary apparatus used to
perform a cryptographic computation in terms of an (electronic) switching circuit.
Then the collection of interconnects between the switching elements,
the \emph{wires} $x_1, x_2, \dots, x_k$,
carry the complete intermediary state information.
Each wire $x_i$ transports an information signal $x_{i}(t)$
as a function of time.
To simplify notation, we consider only
discrete evaluation cycles in time, $t=1,2,\dots, T$,
and binary signal values on the wires throughout this paper.
This model fits CMOS circuit technology already very well,
which is today the dominating technology for
the implementation of electronic (cryptographic) devices.
Hence we can express all $x_{it} = x_{i}(t)$ by elements of the binary field $\F_2$.
A generalization to $n$-ary circuit logic is immediate.
Also the problem of probing analogue signals can be described in
the presented formalism by quantizing and mapping the continuum
of analogue values to an appropriate number of discrete values.
We are now ready to define probing attacks
given an adversary of quantifiable strength. All the definitions are in
accordance with the definitions of differential cryptanalysis,
differential fault attacks, and differential power analysis (of order $q$).

\BDEF
In a \textbf{Probing Attack of Order q, PRA(q),} an adversary
is capable of obtaining the values $(x_{1,t}, x_{2,t}, \dots, x_{q,t})$
on $q$ wires of his choice in a circuit for an arbitrary number of
evaluation cycles $t = 1, 2, \dots, T$.
\label{Def:PRA}
\EDEF

\BDEF
In a \textbf{Differential Probing Attack of Order q, DPRA(q),}
an adversary is capable of obtaining the values
$\V x_t = (x_{1,t}, x_{2,t},\dots, x_{q,t})$ on $q$ wires of a circuit for
an arbitrary number of evaluation cycles $t = 1, 2, \dots, T$.
Furthermore, $\V x_t$ is related to some known information
(\eg the cipher text) and a secret $\V k$
via a set of $\lambda$ equations
$f_i(\V x_t, \V c_t, \V k) = 0$, $1 \le i \le\lambda$.
By collecting multiple different pairs
$(\V x_t, \V c_t)$ and evaluating these equations
the adversary determines some or all bits of the secret $\V k$.
\label{Def:DPRA}
\EDEF

It should be noted that DPRA(q) attacks can be highly efficient.
DPRA(q) works in the cipher text only scenario
and is even more efficient in a known plain text setting.
Probing one bit at a carefully selected position
and collecting the values of a few encryptions can already reveal
the whole key or at least reduce the key space to a
work factor that allows a brute force attack.
Examples for efficient DPRA(1) and DPRA(3) attacks on the AES
are given in~\cite{Schmidt2009ProbingAttackAES}:
The information collected from a single probe during 168 encryptions
reveals the secret key.
In a known plain text setting three probes and 26 cipher texts are already sufficient.
Less efficient attacks on DES, RC5, and public cryptosystems have already been described
in an early work~\cite{Handschuh1999ProbingAttackson}.

Let us now formally consider the state of the $k$ wires of the circuit
at time $t$ as a message word
\begin{align*}
   \V x_t = (x_{1,t}, x_{2,t}, \dots, x_{k,t}) \in \mathcal X = \F_2^k.
\end{align*}
in some message space $\mathcal X$.
If an attacker has access to a small, but carefully selected set of $q$
coordinates of the message word over some period of time,
he will generally be able to extract the secret in a DPRA(q) attack.
Privacy can be preserved if the message is augmented by a
number of $s$ masks
\begin{align*}
   \V m_t = (m_{1,t}, m_{2,t}, \dots, m_{s,t}) \in \mathcal M = \F_2^s.
\end{align*}
\BDEF\label{Def:Mask}
A \textbf{mask} $m_{it}$ is defined to be a value (on a wire) which can be described as
an independent and uniformly distributed binary random variable.
\EDEF
In practical circuit designs a balanced i.i.d. sequence of mask bits
could be generated by a random bit stream generator (RBG).
This sequence is routed on a wire to a
destination circuit element in which the mask bits
are finally
combined with message bits.
The crucial point is the setup of an optimal \emph{masking scheme}
against an adversary with given probing capabilities.
A masking scheme describes the way the vector of masks $m_t$
is combined with the message vector $x_t$ in each evaluation cycle.
In the following we show that the problem of finding an optimal masking scheme
can be expressed as a channel coding problem.

In general we have an encoding function $g$,
which is a map
\begin{align*}
   g: \mathcal X \times \mathcal M &\rightarrow \mathcal Y \subseteq \F_2^N\\
                                (\V x, \V m)&\mapsto \V y = (y_1,y_2,\dots,y_N)
\end{align*}
with $n = s+k$,  $N \ge n$.
The dimension of the image is equal to $n$, because we must be able to
decode the message and we assume that every mask is used.
If no redundancy for integrity protection is introduced we have $N = n$.
In the next sections we consider pure masking schemes ($N = n$).
Finally, in section~\ref{Sec:OTR} we will introduce
tamper-resistant codes which are both, capable of
preserving privacy and integrity ($N>n$).

The following definition is crucial:
%
%%%%%%%%%%%%%%%%%%%%%%%%%%%%%%%%%%%%%%%%%%%%%%%%%%%%%%%%%%%%%%%%%%%%%%%%%%%%%%%%
\BDEF\label{Def:ProbingSecurity}
A circuit is \textbf{probing secure of order q}, we write \textbf{PS(q)},
if for each choice of
indices $i_1,i_2,\dots,i_q$ with $1\le i_1 < i_2 < \dots < i_q \le n$
the condition
\[
    I(X_{1},X_{2},\dots,X_{k}; Y_{i_1},Y_{i_2},\dots,Y_{i_q}) = 0
\]
on the mutual information~\cite{Cover2006ElementsofInformation} holds,
where the message $X_i$ and the masked message $Y_i$ are
represented by discrete random variables
with $p(x_i) = \mathrm{Pr}\{X_i=x_{i,t}\}$
and  $p(y_i) = \mathrm{Pr}\{Y_i=y_{i,t}\}$,
respectively, at each point in time $t$.
\EDEF
In other words, a circuit is PS(q), if it does not leak any information
on the message bits $x_i $ to an attacker,
who can simultaneously probe $q$ wires $y_i$ of his choice
over an arbitrary period of time.
If this condition holds the attack will fail regardless of
whether a simple or differential probing setup is used
according to Definition \ref{Def:PRA} or \ref{Def:DPRA}, respectively.
It should be noted, that Definition~\ref{Def:ProbingSecurity}
is a natural generalization of the
notion of correlation-immunity of a Boolean function, introduced
by Siegenthaler~\cite{Siegenthaler1984Correlation-immunityofnonlinear}.
This definition is also in accordance with the definition
of a power analysis of order $q$, \cf\cite{Mangard2007PowerAnalysisAttacks}.

At this point a natural question arises:
%%%%%%%%%%%%%%%%%%%%%%%%%%%%%%%%%%%%%%%%%%%%%%%%%%%%%%%%%%%%%%%%%%%%%%%%%%%%%%%%
\begin{center}
\begin{minipage}{0.83\textwidth}
 \em What is the lower bound for the number of masks needed to protect
     a circuit with $k$ wires against an adversary who is able to
     mount a probing attack of order $q$?
 \end{minipage}
\end{center}
%%%%%%%%%%%%%%%%%%%%%%%%%%%%%%%%%%%%%%%%%%%%%%%%%%%%%%%%%%%%%%%%%%%%%%%%%%%%%%%%
This information is important for the design of privacy preserving masked
circuits, because \emph{circuit size increases strongly with the number of masks}.
Luckily, experience from physical failure analysis shows that in recent
IC technologies access to single wires
becomes painfully difficult~\cite{Boit2008PhysicalTechniquesChip-Backside}.
Hence the relation between the number of probes and the work factor
of the attack is also a strongly increasing function.
Therefore, in order to design an optimal privacy preserving circuit for a work factor
which is commensurate with the protection period and value of the secret,
it is necessary to know the lower bound for the number of masks.
If the optimal number is known, this value can be used as a target or benchmark
for designs of cryptographic circuits.
The construction of the masking schemes presented in the next section
is based on linear block codes.
These linear schemes are optimal in the sense that they require the
smallest number of masks for a given number of information bits.

%%%%%%%%%%%%%%%%%%%%%%%%%%%%%%%%%%%%%%%%%%%%%%%%%%%%%%%%%%%%%%%%%%%%%%%%%%%%%%%%
\section{Optimal linear masking schemes, OPS-Codes}\label{Sec:OPS}
We shall use the following notation:
$\V G_{ij}$ is a $i\times j$ matrix over $\F_2$,
$\V I_k$ is the $k\times k$ unit matrix,
$\V O_{ij}$ is the $i\times j$ zero matrix,
and $\V 1_k$ is the row vector of $k$ ones.
Furthermore, $\V x^T$ is the transposed vector (matrix) of $\V x$.
The message augmented with the vector of masks is denoted by the row vector
$\V u = (\V x , \V m)$.
We can express any linear masking scheme by
\[
    \V y = \V u \V G,
\]
where $\V G$ is the generator matrix of the code.
As usual the parity check matrix is defined by
\[
    \V H \V y^T = 0,
\]
and we have the relations
\BE
\V H = (\V Q | \V I), \quad \V G = (\V I | \V Q^T)
\label{Eq:GenParRelation}
\EE
choosing the systematic form of the code.

We now define the canonical form of the code for a linear masking scheme.
%
%%%%%%%%%%%%%%%%%%%%%%%%%%%%%%%%%%%%%%%%%%%%%%%%%%%%%%%%%%%%%%%%%%%%%%%%%%%%%%%%
\begin{proposition}
The $n\times n$ generator matrix $\V G$
for a masking scheme with $s$ masks and $n-s$ data values
can be written in the following canonical form:
\BE\label{Eq:ProbingMatrix}
 \V G =
 \left(\begin{array}{c}
  \V I_{n-s} | \V O_{n-s,s} \\
  \hline
  \V P^q_{s,n}
 \end{array}\right)
              \quad
              \textrm{with}
              \quad
 \V P^q_{s,n} = (\V Q_{s,n-s} | \V I_{s}).
\EE
\end{proposition}
\begin{proof}
Using elementary row operations and column interchanges the
generator matrix of a code can be converted into that of
an equivalent code in the given canonical form. \QED
\end{proof}
In the canonical form the first $k = n-s$
coordinates $y_i$ of a code word $\V y$ are
given by the corresponding message bits plus a linear combination of masks,
$y_i = x_i + \sum_{j=1}^{s} m_j Q_{ji}, 1 \le i \le k$.
The last $s$ coordinates are single mask bits.
\BDEF
We shall call $\V P^q_{s,n}$ the \textbf{probing matrix of order q}.
\EDEF
It has $s$ rows (number of masks) and $n=k+s$ columns
(number of masks plus data bits). The probing security of the
corresponding masking scheme is indicated by the superscript $q$.
We now provide some examples.

The classical one-time pad (OTP) encryption scheme of
Gilbert Vernam~\cite{Vernam1919SecretSignalingSystem}
follows in a natural way as one special case.
Every bit of the message is masked with an individual key stream (mask) bit.
OTP provides perfect secrecy under the black-box assumption for the masks.
However, in the setting of physically observable computation the masks
are also subject to the $q$-th order probing attack.
In our new formalism Vernam's masking scheme is defined by the probing matrix
\[
    \V P^1_{k,2k} = (\V I_k | \V I_k).
\]
Obviously the OTP scheme is only PS(1).
An attacker with two needles is able to probe a masked value
and the corresponding mask bit and can therefore compromise the
security of the masking scheme. With one needle no attack is possible.
It is important to point out that there are other PS(1) masking schemes.
One mask is already sufficient to protect an arbitrary
number of $k$ information bits against a PA(1) attack.
The corresponding canonical probing matrix is
\[
    \V P^1_{1,k} = (\V 1_k | 1).
\]
We note that this matrix is identical to the H-Matrix of a single
parity check code for $k$ data bits.
Another trivial case is the construction of the PS(q)
masking scheme for maximum probing security of one data bit.
Clearly we need $s = q$ masks to expand one information bit to $q+1$ shares.
The corresponding canonical probing matrix
\[
    \V P^q_{q,q+1} = (\V 1_q^T | \V I_q)
\]
is the H-Matrix of a repetition code.
Inspired by these observations we can now formulate one of our main results.
%
%%%%%%%%%%%%%%%%%%%%%%%%%%%%%%%%%%%%%%%%%%%%%%%%%%%%%%%%%%%%%%%%%%%%%%%%%%%%%%%%
\begin{theorem}\label{Th:ProbingMatrix}
A linear masking scheme is probing secure of order $q$,
if and only if the probing matrix $\V P^q_{s,n}$
has the property that any $q$ columns are linearly independent.
\end{theorem}
\begin{proof}
Let us denote the $i$th column of the probing matrix
by $P_i$ so that $\V P_{s,n}^q =(P_1,P_2,\dotsc,P_n)$.

To show the necessity of the condition,
assume to the contrary that there are $h$ columns of
$\V P_{s,n}^q$ with $h \leq q$ such that $P_{i_1}+\dotsb +P_{i_h} = \V 0$.
By summing up the corresponding
$y_{i_1},\dotsc,y_{i_h}$
all involved mask bits  will cancel out resulting in
$y_{i_1}+ \dotsb +y_{i_h} = x_{j_1}+\dotsb +x_{j_b}$.
%%%%for some subset $\{j_1,\dotsc,j_b \} \subseteq \{i_1,\dotsc, i_h\}$.
It follows that $I(x_1,\dotsc,x_k; y_{i_1}+\dotsb +y_{i_h})=1$
which implies that $I(x_1,\dotsc,x_k;y_{i_1},\dotsc,y_{i_h})$
cannot be zero, and the masking scheme cannot be
probing secure of order $q$ according to Definition~\ref{Def:ProbingSecurity}.

To show the sufficiency we consider the table of function
values of the function $f:\mathbf{u} \in
\mathbb{F}_2^n \mapsto \V u \V G \in \F_2^n$.
This is a $2^n \times 2n$ matrix denoted by
$[\V u, \V u \V G]$ where $\mathbf{u}$
runs through the elements of $\mathbb{F}_2^n$.
Recall that $\mathbf{u}$ has the form $\mathbf{u} = (\mathbf{x},\mathbf{m})$,
where
$\mathbf{x}=(x_1,\dotsc,x_k)$ is the vector of all $k$ information bits, and
$\mathbf{m}=(m_1,\dotsc,m_s)$ is made up by the $s$ masking bits.
Because of the linearity of $f$,
we can add any two rows of $[\V u, \V u \V G]$,
and the sum will be again a row of $[\V u, \V u \V G]$.
In other words, the rows of the matrix $[\V u, \V u \V G]$
constitute an $n$-dimensional vector space over $\F_2$
denoted by $W$ which is a subspace
of $\mathbb{F}_2^{2n}$.

Let $1 \leq i_1 < i_2 < \dotsb < i_q \leq n$ be arbitrary.
Set $[\V u, \V u \V G] = (A_1,\dotsc,A_n;B_1,\dotsc,B_n)$.
For the proof the following $2^n \times (k+q)$ submatrix $\V M$ of
$[\V u, \V u \V G]$ is essential, where
\begin{equation} \label{Eq:SubmatrixM}
\V M = (A_1,\dotsc,A_k; B_{i_1},\dotsc,B_{i_q}).
\end{equation}
We claim that among the $2^n$ rows of $\V M$ there are exactly $2^{s-q}$ all-zero rows.
We will make use of the standard basis vectors of the vector space $\F_2^n$ given by
\[
\mathbf{e}_1 = (1,0,\dotsc,0), \quad \mathbf{e}_2 = (0,1,\dotsc,0),  \quad  \dotsc, \quad
\mathbf{e}_n = (0,0,\dotsc,1).
\]
Consider the $s$ standard basis vectors $\mathbf{e}_{k+1}, \mathbf{e}_{k+2}, \dotsc,
\mathbf{e}_{k+s}$. The $2^s$ linear combinations (over $\F_2$) of these vectors will
produce \emph{all} row vectors of $\mathbb{F}_2^n$ whose first $k$ coordinates are zero.

Let $V \subset \mathbb{F}_2^n$ be the vector space spanned by
$\mathbf{e}_{k+1},\dotsc,\mathbf{e}_{k+s}$.
Clearly, $\dim(V)= s$.
By hypothesis, the $q$ columns
$P_{i_1},\dotsc,P_{i_q}$ of the probing matrix $\V P_{s,n}^q$ are linearly independent.
This implies that the matrix
$B=(B_{i_1},\dotsc,B_{i_q})$
has rank $q$.
Consider the linear mapping
\[
\phi : \mathbf{v} \in V \mapsto \mathbf{v}B \in \mathbb{F}_2^q.
\]
By elementary linear algebra,
\[
\dim(V) = \dim(\operatorname{ker}(\phi)) + \dim(\operatorname{Im}(\phi)).
\]
Since $\dim(V)=s$, and $\dim(\operatorname{Im}(\phi)) = \operatorname{rank}(B)=q$,
we conclude that
$\dim(\operatorname{ker}(\phi))=s-q$.
Thus there are $2^{s-q}$ vectors $\mathbf{v} \in V$ for which
$\mathbf{v}B = \mathbf{0}$.
It follows that the matrix $\V M$ in \eqref{Eq:SubmatrixM}
has $2^{s-q}$ all-zero rows.

It is now easy to see, that the set $U$ of all vectors of $W$,
which are zero in positions
$1,\dotsc,k, i_1,\dotsc,i_q$,
forms a $(s-q)$-dimensional subspace of the $n$-dimensional vector  space $W$.
Consider the cosets $w+U$, $w\in W$, of the subspace $U$ in $W$.
The distinct cosets of $U$ have all the same cardinality,
namely $2^{s-q}$, the cardinality of $U$. This implies that
among the $2^n$ rows of the matrix $\V M$ in \eqref{Eq:SubmatrixM}
each of the vectors of $\mathbb{F}_2^{k+q}$
occurs exactly $2^{s-q}$ times or, equivalently,
with probability $2^{s-q-n}$. Thus, for
$\mathbf{x}=(x_1,\dotsc,x_k)$ and $\mathbf{z}=(y_{i_1},\dotsc,y_{i_q})$,
we have $p( \mathbf{x},
\mathbf{z} ) = 2^{s-q-n}$, $p(\mathbf{x}) = 2^{-k}$, and $p(\mathbf{z}) = 2^{-q}$.
Since $k=n-s$ it follows that the mutual information,
\cf Definition~\ref{Def:ProbingSecurity},
of $\V x$ and $\V z$ vanishes:
\[
I(\mathbf{x}; \mathbf{z}) = \sum_{\mathbf{x}} \sum_{\mathbf{z}} p(\mathbf{x}, \mathbf{z}) \log_2
\frac{p(\mathbf{x}, \mathbf{z})}{p(\mathbf{x}) p(\mathbf{z})} = 2^{n-s+q} 2^{s-q-n} \log_2
\frac{2^{s-q-n}}{2^{-k} 2^{-q}} = 0.
\]
\QED
\end{proof}
Theorem~\ref{Th:ProbingMatrix} shows that the problem of constructing a
masking scheme for a given order of probing security PS(q) is
equivalent to the problem of constructing a code with given minimum
distance $\dmin = q + 1$. This is immediate, if we recall that the
minimum distance $\dmin$ of a linear code equals the smallest positive
integer $n$ such that there are $n$ columns in the parity check
matrix which are linearly dependent~\cite{MacWilliams2006TheoryofError-Correcting}.

The difference in the case of masking is, that the probing matrix
has to meet the constraints that are usually imposed on the parity check matrix
in the setting of channel coding.
Hence we can state:
\begin{center}
\begin{minipage}{0.83\textwidth}
 \em In the coding theoretical sense
     the problem of preserving privacy in a circuit subject to probing attacks
     is dual to
     the problem of preserving integrity in a circuit subject to fault attacks.
 \end{minipage}
\end{center}
In Section~\ref{Sec:OTR} the duality aspect will be treated further by
considering simultaneous probing and forcing attacks.

Based on Theorem~\ref{Th:ProbingMatrix} we now provide
constructions for codes with different levels of probing security,
which are optimal with respect to the number of masks.
\BDEF\label{Def:OPS}
An \textbf{Optimal Probing Secure Code, OPS(n,k;q),}
is a linear block code of length $n$ and dimension $k$
which provides probing security of order $q$, PS(q),
for $k$ information bits
and has the minimal number $s = n-k$ of mask bits.
\EDEF
By virtue of Theorem~\ref{Th:ProbingMatrix} the existing results
from the field of coding theory can be used for the construction
of good masking schemes for any order of probing security.
The construction of optimal masking schemes for PS(2) and PS(3)
is quite easy:
\begin{itemize}
\item The canonical probing matrix $\V P_{n-k,n}^2$
of the OPS(n,k;2) code is identical
to the parity check matrix of a (shortened)
$[2^s-1,2^s-s-1,3]$ Hamming code.
An example for the OPS(7,4;2) masking scheme can be found in
Appendix~\ref{App:OPS_Example1}.

\item The canonical probing matrix $\V P_{n-k,n}^3$
of the OPS(n,k;3) code is identical
to the parity check matrix of a (shortened)
$[2^{s-1},2^{s-1}-s,4]$ Hsiao code~\cite{Hsiao1970ClassofOptimal}.
An example for the OPS(16,11;3) masking scheme can be found in
Appendix~\ref{App:OPS_Example2}.
\end{itemize}
\begin{table}[h!t]
\centering
\begin{tabular*}{\textwidth}{@{\extracolsep{\fill}}|@{\quad}cc|cccccccccccc@{\quad}|}
\hline
    &    & \multicolumn{12}{c|}{order of probing security $q$} \\
    &    &     1    &    2  &      3 &       4  &      5   &      6 &     7  &     8  &     9  &     10 &     11  &    12 \\
\hline
    & 1  & $\infty$  &       &        &          &          &        &        &        &        &        &         &       \\
    & 2  & $\infty$  &    3  &        &          &          &        &        &        &        &        &         &       \\
    & 3  & $\infty$  &    7  &     4  &          &          &        &        &        &        &        &         &       \\
    & 4  & $\infty$  &   15  &     8  &      5   &          &        &        &        &        &        &         &       \\
    & 5  & $\infty$  &   31  &    16  &      6   &     6    &        &        &        &        &        &         &       \\
$s$ & 6  & $\infty$  &   63  &    32  &      8   &     7    &     7  &        &        &        &        &         &       \\
    & 7  & $\infty$  &  127  &    64  &     11   &     9    &     8  &     8  &        &        &        &         &       \\
    & 8  & $\infty$  &  255  &   128  &     17   &    12    &     9  &     9  &     9  &        &        &         &       \\
    & 9  & $\infty$  &  511  &   256  &     23   &    18    &    11  &    10  &    10  &    10  &        &         &       \\
    & 10 & $\infty$  & 1023  &   512  &\em{34-37}&    24    &    15  &    12  &    11  &    11  &    11  &         &       \\
    & 11 & $\infty$  & 2047  &  1024  &\em{48-60}&\em{35-37}&    23  &    16  &    12  &    12  &    12  &    12   &       \\
    & 12 & $\infty$  & 4095  &  2048  &\em{66-88}&\em{49-61}&    24  &    24  &    14  &    13  &    13  &    13   &   13  \\
\hline
\end{tabular*}
\medskip
\caption{Maximum lengths $n$ of OPS(n,k;q) codes for masking schemes
         given the number of masks $s$ and the order of probing security $q$.}
\label{Tab:MaskNo}
\end{table}
For probing security PS(q) with order $q > 3$
the task of finding the optimal code becomes already nontrivial.
Table~\ref{Tab:MaskNo} shows the upper bound on the length $n$ of the
optimum code given the number of masks $1 \le s \le 12$
and the probing security $1 \le q \le 12$.
The condensed information in Table~\ref{Tab:MaskNo} reflects various
sources---books, articles, as well as online data bases.
For some values only intervals for the
maximum possible length can be given,
because only lower and upper bounds for the minimum
distances of the corresponding codes are known.
These table entries are italicized.

Probing security of order PS(1) for an unbounded number of
information bits can already be achieved by introducing one mask
(Tab.~\ref{Tab:MaskNo}, column 1).
The maximum lengths for OPS(n,k;2) and OPS(n,k;3) masking schemes
are $2^{s}-1$ and $2^{s-1}$, respectively.
These are shown in columns 2 and 3.
The probing matrices correspond to the mentioned Hamming and Hsiao codes, respectively.
The trivial OPS(n,1;n-1) codes, which provide maximum probing security for one data bit,
are found on the diagonal line.
Some comments on a few more selected entries of the table are given:
The entry for $q=4,s=8$ corresponds to the $[17,9,5]$ quadratic
residue code which generates an OPS(17,9;4)
masking scheme. This scheme will be used as an example in the following section,
\cf also Appendix~\ref{App:OPS_Example3}.
The entries for $q=6,s=11$ and $q=7,s=12$ are
the $[23,12,7]$ Golay code $\mathcal G_{23}$ and
the $[24,12,8]$ Golay code $\mathcal G_{24}$, respectively.
It should be noted that the maximal length of a code for a given number of masks
decreases rapidly with increasing order of probing security.
Hence, achieving a high order of probing security in a circuit,
say $q \ge 4$, becomes inefficient in terms of the number of masks.
PS(2) and PS(3) masking schemes, however, are efficient and
may be of high practical relevance.

%%%%%%%%%%%%%%%%%%%%%%%%%%%%%%%%%%%%%%%%%%%%%%%%%%%%%%%%%%%%%%%%%%%%%%%%%%%%%%%%
\section{Information leakage of OPS Masking schemes}\label{Sec:Leakage}
Let us consider a simple probing attack PA(q) on a masked circuit
for increasing order $q=1,2,\dots$.
Obviously, with an increasing number of probes $q$ the circuit
will leak more and more information.
For the moment we do not take into account the possibility that a differential
DPRA(q) attack might reveal the total information already for a small number of probes.
Furthermore, let us consider an intelligent attacker
who follows the optimum strategy in placing the probes.
The incremental information leakage for different masking schemes will
generally look different.
In Fig.~\ref{Fig:InformationLeakage} the information leakage of a circuit
protected by Vernam's OTP scheme ($\times$), \ie one mask per information bit,
an OPS(16,11;3) masking scheme with 5 mask bits ($\bigtriangledown$) and
an OPS(17,9;4) masking scheme with 8 mask bits ($\bigtriangleup$)
are compared.
%%%%%%%%%%%%%%%%%%%%%%%%%%%%%%%%%%%%%%%%%%%%%%%%%%%%%%%%%%%%%%%%%%%%%%%%%%%%%%%%
\begin{figure}[thb]
\centering \scalebox{1.0}{\includegraphics{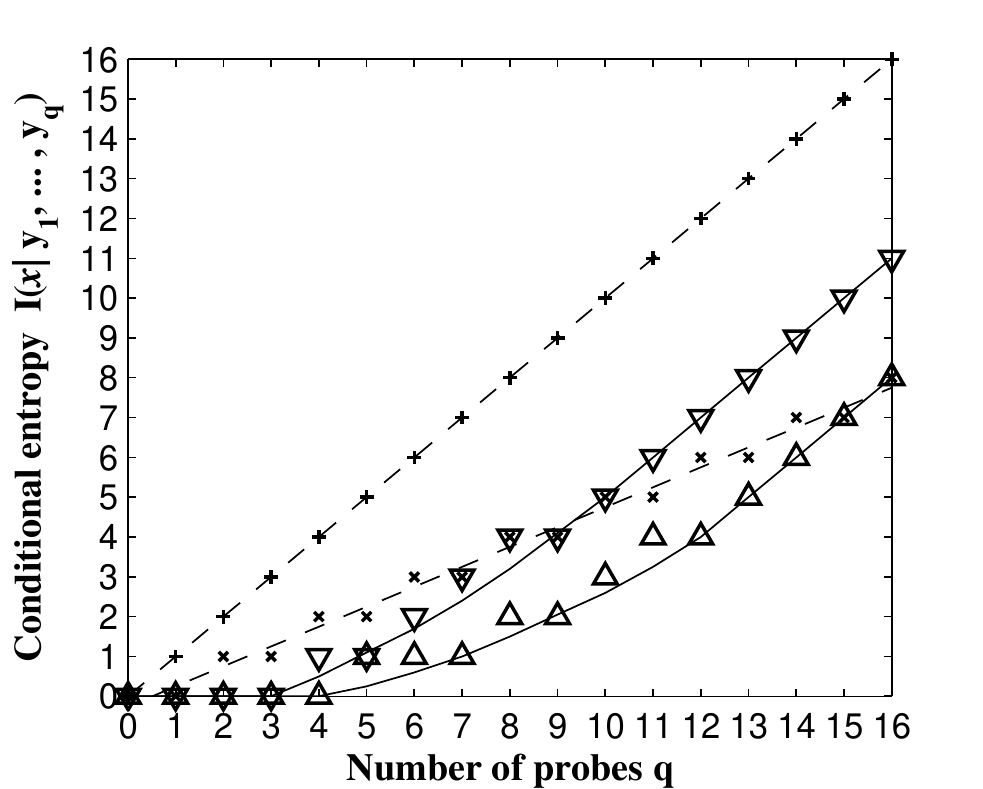}}
\caption{Information leakage of an OPS(16,11;3) code ($\bigtriangledown$) and
         and OPS(17,9;4) code ($\bigtriangleup$)
         in a probing attack with increasing number of needles.
         For comparison an unmasked circuit ($+$)
         and Vernam's masking scheme ($\times$) is also shown.
         The lines serve as a guide to the eye.}
\label{Fig:InformationLeakage}
\end{figure}
%%%%%%%%%%%%%%%%%%%%%%%%%%%%%%%%%%%%%%%%%%%%%%%%%%%%%%%%%%%%%%%%%%%%%%%%%%%%%%%%
For reference purposes the information leakage of an
unmasked circuit ($+$) is also shown. The unmasked circuit
exhibits a constant information leakage rate of one bit per probe.
(We assume that the information on the wires is statistically independent.)
The Vernam masking scheme shows a constant average leakage
rate of 0.5 bit per probe, because the best attack strategy is to probe
one masked wire and the corresponding mask.
In contrast, the OPS masking schemes leak no information up to
the built-in level of probing security,
which is 3 and 4 needles, respectively, in our example.
For an increasing number of probes the information leakage
is still below that of the Vernam scheme.
The OPS(16,11;3) masking scheme reaches the
Vernam leakage rate at 7 probes, while the OPS(17,9;4) scheme
arrives at this rate not until 15 probes.
Asymptotically all OPS leakage rates converge
to a rate of one bit per probe like in an unmasked circuit.
Metaphorically speaking an OPS masking scheme
draws on a private credit to bravely resist attacks with a
moderate number of needles, but the scheme collapses,
if a critical charge is reached.

%%%%%%%%%%%%%%%%%%%%%%%%%%%%%%%%%%%%%%%%%%%%%%%%%%%%%%%%%%%%%%%%%%%%%%%%%%%%%%%%
\section{Simultaneous probing and fault security, OTR-Codes}\label{Sec:OTR}
It is evident to ask for the generalization
of a probing secure circuit to a \emph{tamper-resistant} circuit, which
simultaneously preserves privacy and integrity.
Before further analysis we need to set up a precise fault attack model
for an attacker who can simultaneously mount probing and fault attacks.
An important subclass of physical fault attacks are \emph{forcing attacks}.
In this model we assume that a powerful attacker has full control
on the values of up to $f$ wires of his choice.
An attacker has various physical methods to perform
such a \emph{surgical}, \ie local attack.
A powerful attacker could place $f$ probes and overdrive the signal
values on the wire by applying the appropriate electric potential.
The work factor for such an attack resembles that of a probing attack.
In a weaker, \ie less controlled forcing attack (of statistical nature),
pulses of electromagnetic radiation (e.g. from one or more lasers)
could be used to flip signal values within some radius of influence.
The faulty signal
values will follow some probability distribution
(depending on several physical parameters).
In this attack the work factor for the setup will be smaller.
However, the attacker will have to repeat the attack many
times to generate an error vector that will lead to the intended
information leakage.
If we can assume that only up to $f$ bits are modified,
this attack is also covered by our model (random forcing attack).
\BDEF
In a \textbf{Forcing Attack of Order f, FRA(f),} an adversary
is able to force values $(x_{1,t}, x_{2,t}, \dots, x_{f,t})$
on $f$ wires of his choice in a circuit to 0 or 1
for an arbitrary number of
evaluation cycles $t = 1, 2, \dots, T$.
That means he is able to imprint a (possibly changing)
vector of values $\V e_t=(e_{1,t}, e_{2,t}, \dots, e_{f,t})$
on a subset of $f$ wires in subsequent evaluation cycles $t$.
\EDEF
\BDEF
A circuit is \textbf{forcing secure of order f}, we write \textbf{FRS(f)},
if every error in a forcing attack of order $f$ can be detected.
\label{Def:ForcingSecurity}
\EDEF
It should be noted, that with  Definition~\ref{Def:ForcingSecurity}
we stress the importance of error
detection as a precondition to a reaction on the error.
We do not impose any restriction on the decision
whether the circuit should be constructed such that
the error can be corrected or whether the circuit should enter
a secure state. The latter could be an irreversible transition to a state in
which the secrets are deleted and the circuit is inoperative (self-destruction).

We can now proceed, in the spirit of Definition~\ref{Def:OPS},
to develop an encoding scheme for intermediary state variables,
\eg the wires of a circuit,
which provides security against simultaneous probing and fault attacks.
\BDEF
An \textbf{Optimal Tamper Resistant Code, OTR(n,k,j;f,q),}
is a linear block code of length $n$,
dimension $k$, and $j$ information bits,
which is simultaneously forcing secure of order $f$
and probing secure of order $q$, \ie PS(q) and FRS(f).
The number of redundancy bits $r = n-k$ and
the number of mask bits $s = k-j$ are minimal.
\EDEF
In the second central theorem the canonical form of the OTR code is given
and necessary conditions for the existence are derived.
%
%%%%%%%%%%%%%%%%%%%%%%%%%%%%%%%%%%%%%%%%%%%%%%%%%%%%%%%%%%%%%%%%%%%%%%%%%%%%%%%%
\begin{theorem}
W.l.o.g. the canonical shape of the generator matrix of an OTR code
can be written in the form
\BE
 \V G =
  \left(\begin{array}{c|c|c}
    \V I_{j}    &  \V O      & \V S_{j,r}  \\
    \hline
    \V Q_{s,j}  &  \V I_{s}  & \V R_{s,r}
  \end{array}\right),
\label{Eq:DefOTR}
\EE
where the probing matrix (\cf Eqn.\ \ref{Eq:ProbingMatrix}) is given by
\BE
   \V P^q_{s,n} = (\V Q_{s,j} | \V I_s | \V R_{s,r}).
\label{Eq:ProbMatOTR}
\EE
The code is OTR(n,k,j;f,q), \ie simultaneously PS(q) and FRS(f),
if and only if the following three conditions hold:
\begin{enumerate}
\item The parity check matrix of the code is given by
      \BE
         \V H^f_{r,k} = (\V S^T_{r,j} | \V R^T_{r,s} - \V S^T_{r,j} \V Q^T_{j,s}  | \V I_r).
         \label{Eq:ParMatOTR}
       \EE
\item Any $q$ columns of $\V P^q_{s,n}$ are linearly independent.
\item Any $f$ columns of $\V H^f_{r,k}$ are linearly independent.
\end{enumerate}
\end{theorem}
\begin{proof}
Given the parity check matrix \eqref{Eq:ParMatOTR}
and using \eqref{Eq:GenParRelation}
the corresponding generator matrix
\[
    \V G' = \left(\begin{array}{c|c|c}
                    \V I_{j}   &  \V O      & \V S_{j,r}  \\
                    \hline
                    \V O    &  \V I_{s}  & \V R_{s,r}-\V Q_{s,j} \V S_{j,r}
                    \end{array}\right)
\]
is obtained. We now transform $\V G'$ to an equivalent code $\V G$.
Multiplying the upper slice by $\V Q_{s,j}$ and adding the
result to the lower slice we arrive at \eqref{Eq:DefOTR}.
The second condition follows immediately from Theorem~\ref{Th:ProbingMatrix}.
The third condition follows trivially from the definition
of the minimum distance of a code.
\QED
\end{proof}
The construction of an OTR code for a given number of wires
and a given order of probing and forcing security,
\ie the triplet $(j,q,f)$, is a nontrivial task.
It corresponds to the problem of finding
a triplet of matrices $(\V Q_{s,j}, \V S_{j,r}, \V R_{s,r})$, such that
the corresponding probing \eqref{Eq:ProbMatOTR}
and parity check \eqref{Eq:ParMatOTR} matrices
simultaneously fulfil the constraints on
the minimum number of linearly independent columns.
The competing constraints on the probing and the parity check matrix
indicate again the \emph{duality of the privacy and
the integrity protection problem}, \cf Section~\ref{Sec:OPS}.

It is convenient to recall a theorem of Gilbert and Varshamov:
\begin{theorem} (\textbf{Gilbert \& Varshamov})
Let $l,m,n \in \mathds{N}$ with $l\le m\le n$.
There exists a binary $m\times n$ matrix with the property
that any $l$ columns are linearly independent,
if\,\,\,\,$\sum_{i=0}^{l-1}\binom{n-1}{i} < 2^m$.
\end{theorem}
Evidently, an OTR($n,k,j;f,q$) code fulfils the two Gilbert-Varshamov inequalities
\[
    \sum_{i=0}^{q-1} \binom{n-1}{i} < 2^{k-j}
        \quad\textrm{and}\quad
    \sum_{i=0}^{f-1} \binom{n-1}{i} < 2^{n-k}.
\]
Conversely, it is not obvious whether choosing the smallest
possible values of $s=k-j$ and $r=n-k$ independently
for each inequality does imply the existence of the OTR-code.
However, we observed experimentally that this was a sufficient condition
for all tested small parameters. Usually an even better code can be found.

For moderate values of $(j,q,f)$ OTR codes can
be efficiently constructed
using the following algorithm:
For $r$ and $s$ choose the smallest values according to Table~\ref{Tab:MaskNo}
(or more conservatively according to the Gilbert-Varshamov bound).
Choose the parity check matrix of a $[j+s+r,j+s,f+1]$ code.
Calculate the corresponding generator matrix.
Transform the generator matrix using elementary row operations
to the canonical form by taking care that any $q$ columns in
the parity check matrix are linearly independent.
If this constraint cannot be met
select another parity check matrix
and repeat the procedure. Increasing $s$ or $r$ will generally
increase the number of solutions.
Two examples for OTR codes
can be found
in appendix~\ref{App:OTR_Example1}: OTR(7,4,1;2,2) and
in appendix~\ref{App:OTR_Example2}: OTR(16,11,6;3,3).
It should be noted, that optimal solutions are usually
obtained for smaller values of $r$ and $s$ than indicated by Gilbert's theorem.
The OTR(16,11,6;3,3) code is such an example.

%%%%%%%%%%%%%%%%%%%%%%%%%%%%%%%%%%%%%%%%%%%%%%%%%%%%%%%%%%%%%%%%%%%%%%%%%%%%%%%%
\section{Summary}
We have considered the problem of privacy and integrity protection
in cryptographic circuits in a white-box scenario
for a powerful, yet limited attacker.
By introducing a coding theoretical framework
we have shown that constructing an optimal masking scheme
(as a privacy protection method)
can be considered as the dual problem to finding an optimal error code
(as an integrity protection method).
The new formulation
unifies the known linear masking schemes and
allows us to find
lower bounds for the number of masks needed
to protect a circuit against $q$th order probing attacks.
In this attack scenario the information leakage of
the OPS code based masking
schemes is smaller
than that of Vernam's OTP scheme.
Finally, we considered combined probing and forcing attacks
and derived the structure of optimal linear tamper resistant codes (OTR),
which are eligible to preserve both, privacy and integrity,
in $q$th order probing and $f$th order forcing attacks.
A procedure for the construction of OTR codes has been proposed.

It is immediate that all linear structures of a cryptographic algorithm
can be efficiently protected by OPS and OTR codes.
Although the lower bounds given by the linear constructions
are still applicable for the nonlinear parts of an algorithm
a linear coding scheme generally does not propagate through
a nonlinear operation. Tamper protection of nonlinear
structures will, for example, necessitate the application
of extra masks and of nonlinear codes which are compatible
with the specific nonlinear operation.
This is a target for future analysis.

\clearpage

%%%%%%%%%%%%%%%%%%%%%%%%%%%%%%%%%%%%%%%%%%%%%%%%%%%%%%%%%%%%%%%%%%%%%%%%%%%%%%%
% Bibliography
%%%%%%%%%%%%%%%%%%%%%%%%%%%%%%%%%%%%%%%%%%%%%%%%%%%%%%%%%%%%%%%%%%%%%%%%%%%%%%%%
%\clearpage
\bibliographystyle{plain}
\bibliography{../../../../paperdatabase/sec_innov}
%
%%%%%%%%%%%%%%%%%%%%%%%%%%%%%%%%%%%%%%%%%%%%%%%%%%%%%%%%%%%%%%%%%%%%%%%%%%%%%%%
% Appendices
%%%%%%%%%%%%%%%%%%%%%%%%%%%%%%%%%%%%%%%%%%%%%%%%%%%%%%%%%%%%%%%%%%%%%%%%%%%%%%%%
\appendix
%
%%%%%%%%%%%%%%%%%%%%%%%%%%%%%%%%%%%%%%%%%%%%%%%%%%%%%%%%%%%%%%%%%%%%%%%%%%%%%%%%
\section{Example for OPS(7,4;2) masking scheme}\label{App:OPS_Example1}
A number of $s=3$ mask bits provides probing security PS(2) for data words
of length $k=4$. The canonical probing matrix in the OPS(7,4;2)
masking scheme is a $[7,4,3]$ Hamming code.
\BE
\V P^2_{3,7}=
  \left(\begin{array}{cccc|ccc}
    1 & 1 & 0 & 1 & 1 & 0 & 0 \\
    1 & 0 & 1 & 1 & 0 & 1 & 0 \\
    0 & 1 & 1 & 1 & 0 & 0 & 1
  \end{array}\right).
\EE
In explicit terms, the corresponding generator matrix \eqref{Eq:ProbingMatrix}
induces the masking scheme
\[
    (x_1,\dots,x_4, m_1,m_2,m_3) \mapsto
    (x_1+m_1+m_2,x_2+m_1+m_3,x_3+m_2+m_3,x_4+m_1+m_2+m_3,m_1,m_2,m_3).
\]

%%%%%%%%%%%%%%%%%%%%%%%%%%%%%%%%%%%%%%%%%%%%%%%%%%%%%%%%%%%%%%%%%%%%%%%%%%%%%%%%
\section{Example for OPS(16,11;3) masking scheme}\label{App:OPS_Example2}
A number of $s=5$ mask bits can provide probing security PS(3) for data words
of length $k=11$. The canonical probing matrix in the OPS(16,11;3)
masking scheme is a $[16,11,4]$ Hsiao code~\cite{Hsiao1970ClassofOptimal}.
\BE
\V P^3_{5,16} =
\left(\begin{array}{ccccccccccc|ccccc}
 1 & 1 & 1 & 1 & 1 & 1 & 0 & 0 & 0 & 0 & 1  &  1 & 0 & 0 & 0 & 0 \\
 1 & 1 & 1 & 0 & 0 & 0 & 1 & 1 & 1 & 0 & 1  &  0 & 1 & 0 & 0 & 0 \\
 1 & 0 & 0 & 1 & 1 & 0 & 1 & 1 & 0 & 1 & 1  &  0 & 0 & 1 & 0 & 0 \\
 0 & 1 & 0 & 1 & 0 & 1 & 1 & 0 & 1 & 1 & 1  &  0 & 0 & 0 & 1 & 0 \\
 0 & 0 & 1 & 0 & 1 & 1 & 0 & 1 & 1 & 1 & 1  &  0 & 0 & 0 & 0 & 1
\end{array}\right).
\EE

%%%%%%%%%%%%%%%%%%%%%%%%%%%%%%%%%%%%%%%%%%%%%%%%%%%%%%%%%%%%%%%%%%%%%%%%%%%%%%%%
\section{Example for OPS(17,9;4) masking scheme}\label{App:OPS_Example3}
The $[17,9,5]$ quadratic residue code,
\cf\cite{MacWilliams2006TheoryofError-Correcting},
generates an OPS(17,9;4) masking scheme.
Using the generator polynomial $x^8+x^5+x^4+x^3+1$ the following
canonical probing matrix is obtained.
\BE
\V P^4_{9,17} =
\left(\begin{array}{ccccccccc|cccccccc}
 1 & 0 & 0 & 1 & 1 & 1 & 1 & 0 & 0 &  1 & 0 & 0 & 0 & 0 & 0 & 0 & 0 \\
 0 & 1 & 0 & 0 & 1 & 1 & 1 & 1 & 0 &  0 & 1 & 0 & 0 & 0 & 0 & 0 & 0 \\
 0 & 0 & 1 & 0 & 0 & 1 & 1 & 1 & 1 &  0 & 0 & 1 & 0 & 0 & 0 & 0 & 0 \\
 1 & 0 & 0 & 0 & 1 & 1 & 0 & 1 & 1 &  0 & 0 & 0 & 1 & 0 & 0 & 0 & 0 \\
 1 & 1 & 0 & 1 & 1 & 0 & 0 & 0 & 1 &  0 & 0 & 0 & 0 & 1 & 0 & 0 & 0 \\
 1 & 1 & 1 & 1 & 0 & 0 & 1 & 0 & 0 &  0 & 0 & 0 & 0 & 0 & 1 & 0 & 0 \\
 0 & 1 & 1 & 1 & 1 & 0 & 0 & 1 & 0 &  0 & 0 & 0 & 0 & 0 & 0 & 1 & 0 \\
 0 & 0 & 1 & 1 & 1 & 1 & 0 & 0 & 1 &  0 & 0 & 0 & 0 & 0 & 0 & 0 & 1
\end{array}\right).
\EE

%%%%%%%%%%%%%%%%%%%%%%%%%%%%%%%%%%%%%%%%%%%%%%%%%%%%%%%%%%%%%%%%%%%%%%%%%%%%%%%%
\section{Example for OTR(7,4,1;2,2) tamper resistant code}\label{App:OTR_Example1}
To achieve a forcing security of order 2, FRS(2),
we start with the parity check matrix of a $[7,4,3]$ Hamming code.
The distance of this code is $\dmin = 3$ and the number of redundancy
bits is $r = n - k = 3$.
\BE
\V H =
  \left(\begin{array}{cccc|ccc}
    1 & 1 & 0 & 1 & 1 & 0 & 0 \\
    1 & 0 & 1 & 1 & 0 & 1 & 0 \\
    0 & 1 & 1 & 1 & 0 & 0 & 1
  \end{array}\right).
\EE
As given by Tab.~\ref{Tab:MaskNo} a number of $s = 3$
masks bits is required to achieve PS(2) for an OPS code of length $n = 7$.
Hence a maximum of $j = k-s = 1$
information bits can be protected.
The canonical generator matrix can be easily constructed
by applying elementary row operations:
\BE
\V G =
  \left(\begin{array}{c|ccc|ccc}
    1 & 0 & 0 & 0 & 1 & 1 & 0 \\
    \hline
    1 & 1 & 0 & 0 & 0 & 1 & 1 \\
    1 & 0 & 1 & 0 & 1 & 0 & 1 \\
    0 & 0 & 0 & 1 & 1 & 1 & 1
  \end{array}\right).
\EE
It is immediate that any two columns in the probing matrix
(lower part of $\V G$) are linearly independent.
Hence this OTR code is PS(2).

%%%%%%%%%%%%%%%%%%%%%%%%%%%%%%%%%%%%%%%%%%%%%%%%%%%%%%%%%%%%%%%%%%%%%%%%%%%%%%%%
\section{Example for OTR(16,11,6;3,3) tamper resistant code}\label{App:OTR_Example2}
In this nontrivial example
we use a minimum weight Hsiao code ($\dmin=4$)
of length $n=16$ and dimension $k=11$
as a starting point to achieve FRS(3),
\BE
\V H =
  \left(\begin{array}{ccccccccccc|ccccc}
  1 & 1 & 1 & 1 & 1 & 1 & 0 & 0 & 0 & 0 & 1 &  1 & 0 & 0 & 0 & 0\\
  1 & 1 & 1 & 0 & 0 & 0 & 1 & 1 & 1 & 0 & 1 &  0 & 1 & 0 & 0 & 0\\
  1 & 0 & 0 & 1 & 1 & 0 & 1 & 1 & 0 & 1 & 1 &  0 & 0 & 1 & 0 & 0\\
  0 & 1 & 0 & 1 & 0 & 1 & 1 & 0 & 1 & 1 & 1 &  0 & 0 & 0 & 1 & 0\\
  0 & 0 & 1 & 0 & 1 & 1 & 0 & 1 & 1 & 1 & 1 &  0 & 0 & 0 & 0 & 1
 \end{array}\right).
\EE
From Tab.\ \ref{Tab:MaskNo} we see that $s=5$ masks are necessary
to secure $n=16$ bits against a probing attack of order 3.
Applying elementary row operations the generator matrix of
an equivalent PS(3)-secure code can be constructed:
\BE
\V G =
  \left(\begin{array}{cccccc|ccccc|ccccc}
  1 & 0 & 0 & 0 & 0 & 0 &  0 & 0 & 0 & 0 & 0 &  1 & 1 & 1 & 0 & 0  \\
  0 & 1 & 0 & 0 & 0 & 0 &  0 & 0 & 0 & 0 & 0 &  1 & 1 & 0 & 1 & 0  \\
  0 & 0 & 1 & 0 & 0 & 0 &  0 & 0 & 0 & 0 & 0 &  1 & 1 & 0 & 0 & 1  \\
  0 & 0 & 0 & 1 & 0 & 0 &  0 & 0 & 0 & 0 & 0 &  1 & 0 & 1 & 1 & 0  \\
  0 & 0 & 0 & 0 & 1 & 0 &  0 & 0 & 0 & 0 & 0 &  1 & 0 & 1 & 0 & 1  \\
  0 & 0 & 0 & 0 & 0 & 1 &  0 & 0 & 0 & 0 & 0 &  1 & 0 & 0 & 1 & 1  \\
  \hline
  1 & 1 & 1 & 1 & 1 & 0 &  1 & 0 & 0 & 0 & 0 &  1 & 0 & 0 & 1 & 0  \\
  1 & 0 & 1 & 0 & 0 & 1 &  0 & 1 & 0 & 0 & 0 &  1 & 1 & 0 & 1 & 1  \\
  1 & 0 & 0 & 1 & 1 & 1 &  0 & 0 & 1 & 0 & 0 &  0 & 0 & 1 & 1 & 1  \\
  1 & 1 & 1 & 0 & 1 & 0 &  0 & 0 & 0 & 1 & 0 &  0 & 1 & 1 & 0 & 1  \\
  1 & 1 & 0 & 1 & 0 & 1 &  0 & 0 & 0 & 0 & 1 &  1 & 1 & 1 & 0 & 0
\end{array}\right).
\EE
This OTR code can secure $j = k - s = 6$ bits of information
simultaneously against FRA(3) and PA(3) attacks.

%%%%%%%%%%%%%%%%%%%%%%%%%%%%%%%%%%%%%%%%%%%%%%%%%%%%%%%%%%%%%%%%%%%%%%%%%%%%%%%%
%%%%%%%%%%%%%%%%%%%%%%%%%%%%%%%%%%%%%%%%%%%%%%%%%%%%%%%%%%%%%%%%%%%%%%%%%%%%%%%%
\end{document}